\documentclass[conference]{IEEEtran}

\usepackage{url,color,cite}
\usepackage[pdftex]{graphicx}
\usepackage[cmex10]{amsmath}
\usepackage{amssymb}
\usepackage{amsthm}
\usepackage{float}
\usepackage{epstopdf}
\usepackage{hyperref}
\usepackage{algorithm,algpseudocode}
\usepackage{algpseudocode}
\usepackage{pifont}
\usepackage{subcaption}

\newtheorem{theorem}{Theorem}

\ifCLASSINFOpdf
\fi
\IEEEoverridecommandlockouts
\begin{document}
\title{Optimal Energy Allocation For Delay-Constrained Traffic Over Fading Multiple Access Channels}

\author{\thanks{This publication was made possible by NPRP grant $\#$~NPRP $7-923-2-344$ from the Qatar National Research Fund (a member of Qatar Foundations). The statements made herein are solely the responsibility of the authors.}
\IEEEauthorblockN{Antonious M. Girgis\IEEEauthorrefmark{1},
        Amr El-Keyi\IEEEauthorrefmark{2} and Mohammed Nafie\IEEEauthorrefmark{2}\IEEEauthorrefmark{3}
        }
    \IEEEauthorblockA{\IEEEauthorrefmark{1}Wireless Intelligent Networks Center (WINC), Nile University, Cairo, Egypt\\
    \IEEEauthorrefmark{2} Department of Systems and Computer Engineering, Carleton University, Ottawa, Canada\\
        \IEEEauthorrefmark{3} Electronics and Communications Dept., Faculty of Engineering, Cairo University, Giza, Egypt\\
        Email: \{a.mamdouh@nu.edu.eg, amr.elkeyi@sce.carleton.ca, mnafie@nu.edu.eg\}
    }
}

\maketitle

\begin{abstract}

In this paper, we consider a multiple-access fading channel where $N$ users transmit to a single base station (BS) within a limited number of time slots.
We assume that each user has a fixed amount of energy available to be consumed over the transmission window.
We derive the optimal energy allocation policy for each user that maximizes the total system throughput under two different assumptions on the channel state information. First, we consider the offline allocation problem where the channel states are known a priori before transmission. We solve a convex optimization problem to maximize the sum-throughput under energy and delay constraints. Next, we consider the online allocation problem, where the channels are causally known to the BS and obtain the optimal energy allocation via dynamic programming when the number of users is small. We also develop a suboptimal resource allocation algorithm whose performance is close to the optimal one. Numerical results are presented showing the superiority of the proposed algorithms over baseline algorithms in various scenarios.
\end{abstract}
\IEEEpeerreviewmaketitle

\section{Introduction}
Wireless communication channels are characterized by their time-varying fading nature that has significant effect on the performance of wireless networks. Various algorithms have been proposed to design efficient resource allocation schemes that optimize the system performance over fading channels, e.g., by minimizing the transmission power, minimizing the delay, or maximizing the system throughput. Resource allocation over fading channels has been studied for point-to point communication in different contexts, e.g.,~\cite{goldsmith1997capacity,fu2006optimal,negi2002delay,wang2014power,lee2009energy}. In~\cite{goldsmith1997capacity} the expected Shannon capacity for fading channels was obtained when the channel state information (CSI) is known causally at the transmitter and the receiver. Furthermore, it has been shown that the ``water-filling'' algorithm achieves the maximum expected capacity. The authors of~\cite{lee2009energy} considered the problem of minimizing the expected energy to transmit a single packet over a fading channel subject to a hard deadline. In~\cite{fu2006optimal}and~\cite{negi2002delay}, a dynamic program formulation was proposed to maximize a general throughput function under constraints on the delay and the amount of energy available at the transmitter. In~\cite{wang2014power}, the work of~\cite{fu2006optimal} was extended to energy harvesting systems where the transmitter has causal CSI.


The capacity region of the multiple access channel (MAC) has been studied in various settings, see for example~\cite{yu2004iterative,tse1998multiaccess,hanly1998multiaccess,rezki2014capacity,devassy2015finite,budkuley2014jamming}. In~\cite{yu2004iterative}, the capacity region of the Gaussian multiple-input multiple-output (MIMO) MAC was characterized. The authors of~\cite{yu2004iterative} proposed an iterative water-filling algorithm to obtain the optimal transmit covariance matrices of the users that maximize the weighted sum capacity. In~\cite{tse1998multiaccess} the capacity region of the fading MAC was characterized by Tse and Hanly. Furthermore, the power allocation policy that maximizes the long-term achievable rates subject to average power constraints for each user was introduced. In~\cite{hanly1998multiaccess}, Hanly and Tse introduced an information-theoretic characterization of the capacity region of the fading MAC with delay constraints. In addition, they  provided the optimal power allocation policy that achieves the delay-limited capacity. In~\cite{wang2015iterative}, Wang developed the optimal energy allocation strategy for the fading MAC with energy harvesting nodes by assuming that the CSI is \textit{non-causally} known before the beginning of transmission. In~\cite{caire2004variable}, the capacity region of the fading MAC with power constraint on each codeword was investigated. However, the authors of~\cite{caire2004variable} focused their work on the low signal-to-noise ratio (SNR) regime where they showed that the one-shot power allocation policy is asymptotically optimal.

In this paper, we consider a system composed of multiple users transmitting to a single base station (BS) over a fading MAC. The transmission occurs over a limited time duration in which each user has a fixed amount of energy. Some motivating scenarios and applications for this system model are introduced in~\cite{hanly1998multiaccess,fu2006optimal,negi2002delay}, e.g., satellites, remote sensors, and cellular phones with limited amount of energy transmitting delay-sensitive data to a single receiver. We develop energy allocation strategies to maximize the expected sum-throughput of the fading MAC subject to hard deadline and energy constrains. First, we consider the offline allocation problem in which the channel states are known a priori to the BS. We show that the optimal solution of this problem can be obtained via the iterative water filling algorithm. Next, a dynamic program formulation is introduced to obtain the optimal online allocation policy when only causal CSI is available at the BS. Since the computational complexity of the optimal online policy increases exponentially with the number of users,  we develop a suboptimal solution for the online allocation problem by exploiting the proposed offline allocation policy. Moreover, we investigate numerically the performance of the proposed policies and compare them with the equal-energy allocation and the one-shot energy allocation policy of~\cite{caire2004variable}.

The rest of the paper is organized as follows. In Section~\ref{system}, we present the system model and formulate the maximum sum-throughput optimization problem. The offline energy allocation is introduced in Section~\ref{offline}. We study the online allocation in Section~\ref{online}, where dynamic programming is utilized to obtain the optimal policy and a suboptimal policy with reduced computational complexity is proposed. In Section~\ref{results}, we present our numerical results and compare the performance of different policies in various scenarios. Finally, we conclude the paper in Section~\ref{conclusion}.

\section{SYSTEM MODEL}
\label{system}
We consider a discrete-time MAC as shown in Fig.~\ref{F1:system}, where $N$ users communicate with a single BS in a slotted wireless network. We assume a flat-fading channel model in which the channel gain of each user is constant over the duration of the time slot and changes independently from time slot to another according to a known continuous distribution. Thus, the received signal by the BS at time slot $t$ is given by
\begin{equation}
y_t=\sum_{i=1}^{N} \sqrt{h_t^{\left(i\right)}} x_t^{\left(i\right)}+n_t
\end{equation}
where $n_t$ is a zero-mean white Gaussian noise with variance $\sigma^2$, and $x_t^{\left(i\right)}$ is the transmitted signal of user $i$ at time slot $t$. The channel gain between the $i$th user and the BS at time slot $t$ is denoted by $h_t^{\left(i\right)}$, where the channel gains of each user $\left\{h_t^{\left(i\right)}\right\}$, $i\in\lbrace 1,\cdots,N\rbrace$, are independent identically distributed with the cumulative distribution function (CDF) $F_H^{\left(i\right)}\left(x\right)$. Let $E_i$ denote the maximum amount of energy that can be expended by user $i$ during $T$ time slots, where $T$ denotes the transmission window in which each user must transmit his data. Let $\mathcal{N}=\left\{1,\cdots,N\right\}$ denote the set of users communicating with the BS, and $\mathcal{T}=\left\{1,\cdots,T\right\}$ denote the set of the time slots during which communication occurs. Our goal is to maximize the  sum-throughput of the MAC over the transmission window under constraints on the available energy for each user.

Let $e_t^{\left(i\right)}$ denote the consumed energy by the $i$th user at time slot $t$. Hence, the maximum achievable sum-throughput of the MAC at time slot $t$, when the channel gains of all users at time slot $t$ are known, is given by~\cite{cover2012elements}
\begin{equation}\label{eqn7}
R\left(\mathbf{e}_t,\mathbf{h}_t\right)=\tau W\log_2\left(1+\frac{1}{\tau N_o}\sum _{i=1}^{N} h_t^{\left(i\right)}e_t^{\left(i\right)}\right)
\end{equation}
where $W$ and $\tau$ are the channel bandwidth, and the time slot duration, respectively, and $N_o=W\sigma^2$ is the noise power in watts.\footnote{Note that the successive cancellation decoding strategy is the optimal decoding scheme that achieves the maximum sum-throughput of the MAC~\cite{hanly1998multiaccess}} In \eqref{eqn7}, $\mathbf{h}_t=\left[ h_t^{\left(1\right)},\cdots,h_t^{\left(N\right)} \right]$ and $\mathbf{e}_t=\left[e_t^{\left(1\right)},\cdots,e_t^{\left(N\right)}\right]$ are the channel gains vector and the consumed energy vector of all users at time slot $t$, respectively. Let $\mathcal{E}_t^{\left(i\right)}$ be the available energy for user $i$ at time slot $t$. Thus, the evolution of the energy queue of the $i$th user is given by
\begin{equation}\label{energylevel}
\begin{aligned}
\mathcal{E}_{t+1}^{\left(i\right)}=\mathcal{E}_{t}^{\left(i\right)}- e_t^{\left(i\right)} &\qquad t=1,\cdots,T-1
\end{aligned}
\end{equation}
where the initial state of the energy queue is $\mathcal{E}_{1}^{\left(i\right)}=E_i$. In addition, the energy vector $\mathcal{E}_t=\left[\mathcal{E}_{t}^{\left(1\right)},\cdots,\mathcal{E}_{t}^{\left(N\right)}\right]$ represents the energy levels of all users at time slot $t\in\mathcal{T}$.

We aim to get the energy allocation policy for each user $i\in \mathcal{N}$ to maximize the expected sum-throughput of the MAC over a deadline of $T$ slots. Towards this objective, we formulate the following optimization problem:
\begin{equation}\label{eqn1}
\begin{aligned}
\max_{\mathbf{e}_1,\cdots,\mathbf{e}_T} & \quad\mathbb{E}\left\{\sum_{t=1}^{T} R\left(\mathbf{e}_t,\mathbf{h}_t\right)\right\}&&\\
\text{s.t.} &\quad \sum_{t=1}^{T}e_t^{\left(i\right)}= E_i & i&\in \mathcal{N}&\\
&\quad\mathbf{e}_t \succeq \mathbf{0} & t&\in \mathcal{T}&
\end{aligned}
\end{equation}
 where $\mathbf{0}$ denotes a row vector whose elements are equal to zero, $\mathbb{E}$ denotes the expectation with respect to the channel vectors $\mathbf{h}_t$, $t\in \mathcal{T}$, and the maximization is over all feasible energy allocation policies. In the following sections, we first study the offline allocation policy in which the channel gains of all users are known a priori for $T$ time slots. Next, we study different online allocation policies that maximize the expected sum-throughput of the MAC when only causal CSI is available.

\begin{figure}
\centering
\includegraphics[width=7 cm, height=6 cm]{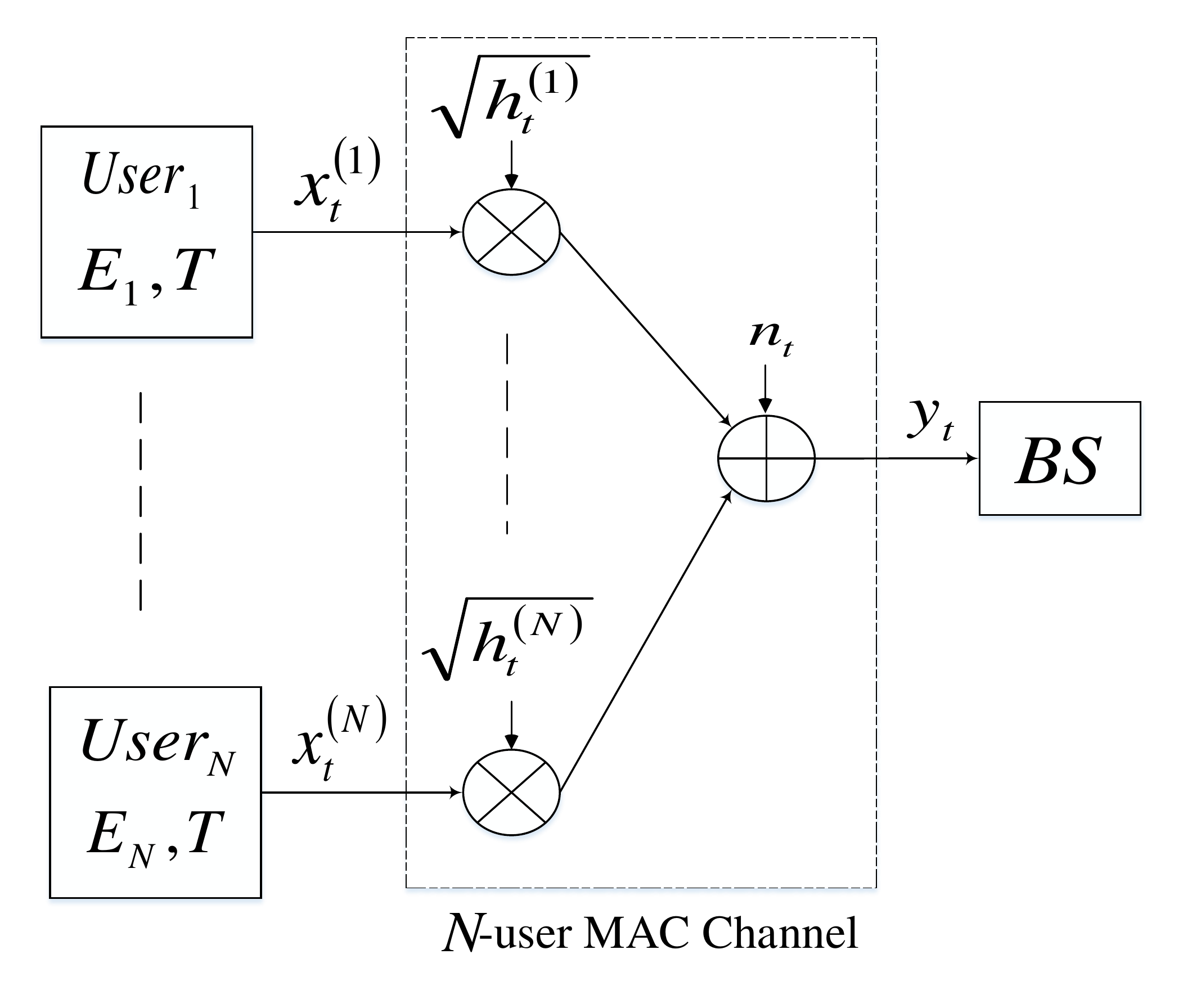}
\caption{System model}
\label{F1:system}
\end{figure}

\section{offline energy allocation}\label{offline}

In this section, we introduce the optimal offline energy allocation policy  when the channel vectors $\mathbf{h}_t$, $t\in \mathcal{T}$, are non-causally known to the BS and the users at the beginning of the transmission. Since the channel vectors $\mathbf{h}_t$, $t\in \mathcal{T}$ are a priori known, the optimization problem~\eqref{eqn1} can be reformulated as a deterministic optimization problem
\begin{equation} \label{eqn2}
\begin{aligned}
\max_{\mathbf{e}_1,\cdots,\mathbf{e}_T} & \quad\sum_{t=1}^{T} R\left(\mathbf{e}_t,\mathbf{h}_t\right)&&
\end{aligned}
\end{equation}
subject to the same constraints of the optimization problem~\eqref{eqn1}.

\begin{theorem}\label{Th1}
The optimal offline transmission policy for the users is obtained by solving the following equations
\begin{eqnarray}
  e_t^{\left(i\right)} &\!\!\!\!=\!\!\!\!& \left(\gamma_{o}^{\left(i\right)}-\gamma_t^{\left(i\right)}\right)^{+} \qquad  \forall i\in \mathcal{N}, t\in \mathcal{T} \label{eqn3}\\
  \sum_{t=1}^{T}e_t^{\left(i\right)} &\!\!\!\!=\!\!\!\!&  E_i  \qquad\qquad \quad\qquad\; \forall i\in \mathcal{N} \label{eqn4}\\
  \gamma_t^{\left(i\right)} &\!\!\!\!=\!\!\!\!& \frac{\tau N_o+\sum_{n\neq i}h_t^{\left(n\right)}e_t^{\left(n\right)}}{h_t^{\left(i\right)}}  \label{last_IWF}
\end{eqnarray}
where $\gamma_{o}^{\left(i\right)}$ is a threshold value obtained by substituting from~\eqref{eqn3} into~\eqref{eqn4}, and $\left(x\right)^{+}=\max\left(0,x\right)$.
\end{theorem}
\begin{proof}
Refer to the Appendix.
\end{proof}

In the single user case, i.e., $N=1$, the optimal offline policy in Theorem~\ref{Th1} is the conventional water-filling algorithm~\cite{goldsmith1997capacity}, where the noise to the channel gain ratio at each time slot $t$ determines the amount of energy allocated to the time slot $t$. In case of the multiple users, i.e., $N>1$, we note that the energy allocation policy of the $i$th user for a given energy allocation of the other users $\mathbf{e}^{\left(i\right)}=\left[e_1^{\left(i\right)},\cdots,e_T^{\left(i\right)}\right]$,  is also obtained via the water-filling algorithm. However, in this case, the interference signals of the other users $\sum_{n\neq i}h_t^{\left(n\right)} e_t^{\left(n\right)}$ at each time slot $t$ are considered as   noise. Hence, the energy allocation policy of the $i$th user is significantly affected by the energy allocation policy of the other users, where the allocated energy for the $i$th user in time slot $t$ depends on $\gamma_t^{\left(i\right)}$ which represents the ratio between the interference-plus-noise power and the channel gain of the $i$th user at time slot $t$.
\begin{algorithm}
\caption{Iterative water-filling (IWF) algorithm }
\label{IWF}
\begin{algorithmic}[1]
\State \textbf{Initialization:} $\mathbf{e}_t=\mathbf{0}$,  $\forall\ t\in \mathcal{T}$
\For {$l=1$ to $L_{\max}$}
\For{$i=1$ to $N$}
\State Let $\gamma_t^{\left(i\right)}=\frac{\tau\sigma^2+\sum_{n\neq i}h_t^{\left(n\right)}e_t^{\left(n\right)}}{h_t^{\left(i\right)}}$, $\forall\ t\in \mathcal{T}$
\State $e_t^{\left(i\right)}=\left(\gamma_{o}^{\left(i\right)}-\gamma_t^{\left(i\right)}\right)^{+}$,  $\forall\ t\in \mathcal{T}$
\State $\sum_{t=1}^{T}e_t^{\left(i\right)}=E_i$
\EndFor
\EndFor
\end{algorithmic}
\end{algorithm}

Note that a closed-form expression for the optimal solution introduced in Theorem~\ref{Th1} can not be found. Nevertheless, the optimal solution can be obtained by applying the iterative water filling algorithm (IWF) described in Algorithm~\ref{IWF} to iteratively solve equations~\eqref{eqn3}--\eqref{last_IWF} where $L_{\max}$ is the maximum number of iterations.  In each iteration, the IWF algorithm successively updates the optimal energy allocation of each user using the water-filling algorithm while assuming that the allocation policy of the other users are fixed. Hence, at each iteration the algorithm tries to maximize the objective function of the problem~\eqref{eqn2} by adapting the energy allocation of a single user while considering the signals of the other users $\sum_{n\neq i}h_t^{\left(n\right)} e_t^{\left(n\right)}$ as noise. Since the objective function is monotonically increasing in the energy allocation policy of each user $\mathbb{e}^{\left(i\right)}$, the objective function cannot decrease after any iteration. As a result,  the IWF solution approaches the optimal solution of problem~\eqref{eqn2} as the number of iterations $L_{\max}$ increases where $L_{\max}$ determines the error tolerance.

The IWF algorithm was applied in~\cite{yu2004iterative} to find the optimal transmit covariance matrices of the users that achieve the boundary of the Gaussian MIMO-MAC capacity. In a similar manner to~\cite{yu2004iterative}, we can assume the channel gains of the $i$th user over the time window ($h_1^{\left(i\right)},\cdots,h_T^{\left(i\right)}$) as effective channel gains of $T$ transmit antennas of the $i$th user. Therefore the results of the IWF algorithm obtained in~\cite{yu2004iterative} can be applied here.
\begin{theorem}
For a finite number of iterations, the IWF algorithm described in Algorithm~\ref{IWF} converges to the optimal allocation policy which is the solution of the optimization problem in~\eqref{eqn2}. Furthermore, the IWF algorithm achieves a sum-throughput lower than the optimal within $\frac{\left(N-1\right)T}{2}$ nats after a single iteration.
\end{theorem}
\begin{proof}
See Theorem~$4$ and Theorem~$5$ in~\cite{yu2004iterative}.
\end{proof}

 \section{Online energy allocation}
\label{online}
In this section, we assume that the channel vector $\mathbf{h}_t$ is causally known to the BS and the users at the beginning of time slot $t$ while future channel states are not known. Let $X_t=\left(\mathcal{E}_t,\mathbf{h}_t\right)$ denote the state of the system which is comprised of the channel gains and the energy levels of all users at time slot $t$. We aim to obtain the energy allocation policy $\mathcal{G}^*=\left[\mathbf{e}^*_1\left(X_1\right),\cdots,\mathbf{e}^*_T\left(X_T\right)\right]$ that maximizes the expected sum-throughput of the MAC within a duration of $T$ slots by sequentially solving the optimization problem in~\eqref{eqn1}. The optimal energy allocation policy $\mathcal{G}^*$ can be obtained by formulating the optimization problem in~\eqref{eqn1} as a finite horizon dynamic program (DP) that can be described by the following two equations
\begin{subequations}
\label{eqn6}
\begin{align}
\label{eqn6:1}
U_T\left(\mathcal{E}_T,\mathbf{h}_T\right)&=R\left(\mathcal{E}_T,\mathbf{h}_T\right)\\
U_t\left(\mathcal{E}_t,\mathbf{h}_t\right)=&\underset{\mathbf{0}\preceq \mathbf{e}_t\preceq \mathcal{E}_t}{\max} R\left(\mathbf{e}_t,\mathbf{h}_t\right)+\overline{U}_{t+1}\left(\mathcal{E}_{t}-\mathbf{e}_t\right) \,\forall 1\leq t<T
\end{align}
\end{subequations}
where $\overline{U}_{t+1}\left(\mathcal{E}\right)=\mathbb{E} \left\{ U_{t+1}\left(\mathcal{E},\mathbf{h}\right)\right\}$. The equations in~\eqref{eqn6} are  Bellman's equations of the finite horizon DP~\cite{bertsekas1995dynamic}, where $\overline{U}_{t+1}\left(\mathbf{\mathcal{E}}\right)$ is the maximum expected sum-throughput that can be obtained during the remaining $T-t$ slots given that the energy levels of all users is $\mathcal{E}$. Note that the optimal policy is a vector of functions mapping the current state of the system (the channel gains and the energy levels) to an amount of energy determined for each user. In~\eqref{eqn6:1} the users transmit with all available energy $\mathcal{E}_{T}$ to maximize the total sum-throughput at the last time slot $T$. On the other hand, for time slots $t=1,\cdots,T-1$ there is a tradeoff between the current reward $R\left(\mathbf{e}_t,\mathbf{h}_t\right)$ and the expected future reward $\overline{U}_{t+1}\left(\mathbf{\mathcal{E}_t-\mathbf{e}_t}\right)$. Hence, the optimal energy allocated for each user at time slot $t$, is determined by maximizing the current throughput plus the expected future throughput.

\subsection{The optimal policy}
The optimal allocation policy $\mathcal{G}^*$ is obtained recursively by solving Bellman's equation in~\eqref{eqn6} at each time slot $t\in \mathcal{T}$, where $\overline{U}_{t+1}\left(\mathcal{E}\right)$ is computed backwards in time. However, we can not get a closed-form expression for the expected reward function $\overline{U}_{t+1}\left(\mathcal{E}\right)$ even in the single user case. Therefore, $\overline{U}_{t+1}\left(\mathcal{E}\right)$ is computed numerically using the discretization method~\cite{bertsekas1995dynamic}.

\subsection{Suboptimal policy}~\label{sub}
The computational complexity required to solve~\eqref{eqn6} numerically grows exponentially with the number of users~\cite{bertsekas1995dynamic}. In order to alleviate this problem, the one-shot energy allocation policy was introduced in~\cite{negi2002delay} and~\cite{caire2004variable} to solve~\eqref{eqn6} efficiently.  The one-shot energy allocation policy arises from the linear approximation of the throughput function, i.e.,
\begin{equation}
  R\left(\mathbf{e}_t,\mathbf{h}_t\right)\approx  \frac{1}{\sigma^2}\sum _{i=1}^{N} h_t^{\left(i\right)}e_t^{\left(i\right)}.
\end{equation}
Note that the linear approximation is acceptable in the wideband regime, i.e., when  $W\rightarrow\infty$, and/or when all users transmit at low SNR, where the transmit SNR of the $i$th user is given by $\text{SNR}_i=\frac{E_i h_{o}^{\left(i\right)}}{\tau N_o}$, $ h_o^{\left(i\right)}=\int_{0}^{\infty} x dF_H^{\left(i\right)}\left(x\right)$. Hence,  Bellman's equations  can be restated as follows
\begin{subequations}
\label{eqn9}
\begin{align}
\label{eqn9:1}
\tilde{U}_T\left(\mathcal{E}_T,\mathbf{h}_T\right)&=\sum _{i=1}^{N} h_T^{\left(i\right)}\mathcal{E}_T^{\left(i\right)}\\
\tilde{U}_t\left(\mathcal{E}_t,\mathbf{h}_t\right)=&\underset{\mathbf{0}\preceq \mathbf{e}_t\preceq \mathcal{E}_t}{\max} \sum _{i=1}^{N} h_t^{\left(i\right)}e_t^{\left(i\right)}+\overline{\tilde{U}}_{t+1}\left(\mathcal{E}_{t}-\mathbf{e}_t\right), 1\leq t<T
\end{align}
\end{subequations}
By applying the DP recursion backward in time from the time slot $t=T$ to the time slot $t=1$, the expected reward function $\overline{\tilde{U}}_{t}\left(\mathcal{E}\right)$  for $t\in\mathcal{T}$  can be computed as follows
\begin{equation}
\overline{\tilde{U}}_{t}\left(\mathcal{E}\right)=\sum_{i=1}^{N}\mathcal{E}^{\left(i\right)}\nu_t^{\left(i\right)}
\end{equation}
Furthermore, the one-shot energy allocation policy which solves~\eqref{eqn9} is given by
\begin{equation}
e_t^{\left(i\right)}=\left\{\begin{array}{ll}
\mathcal{E}_t^{\left(i\right)} & \text{if}\ \ h_t^{\left(i\right)}>\nu_{t+1}^{\left(i\right)}\\
0   & \text{if}\ \ h_t^{\left(i\right)}\leq\nu_{t+1}^{\left(i\right)}
\end{array}
\right.,\ \forall\ i\in\mathcal{N}
\end{equation}
where $\nu_{t-1}^{\left(i\right)}=\mathbb{E}_{h_t^{\left(i\right)}}\left\{\max\lbrace h_t^{\left(i\right)},\nu_t^{\left(i\right)}\rbrace \right\}$, $\nu_{T}^{\left(i\right)}=h_o^{\left(i\right)}$, and $\nu_{T+1}^{\left(i\right)}=0$. The one-shot energy allocation policy allocates the available energy for the $i$th user $E_i$  to the earliest time slot $t\in\mathcal{T}$ that has a channel gain $h_t^{\left(i\right)}>\nu_{t+1}^{\left(i\right)}$.  We refer to~\cite{caire2004variable} for more insights and details.

Next, we develop a low-complexity suboptimal solution to solve the recursive DP introduced in~\eqref{eqn6}. We show through numerical simulations that the performance of the suboptimal algorithm is close to that of the optimal policy when the number of users is small. The suboptimal solution is obtained by applying the certainty equivalent controller (CEC) scheme (see Chapter~$6$ in~\cite{bertsekas1995dynamic}), in which the following three steps are applied at each time slot $t$:
\begin{enumerate}
\item \textbf{Certainty step}: We replace all uncertain variables with their means. Hence we assume that the future channel gain of each user over the remaining $T-t$ slots is equal to its mean, i.e., $\mathbf{h}_k=\mathbf{h}_o$  for $\ k=t+1,\cdots,T$  where $\mathbf{h}_o=\left[h_o^{\left(1\right)},\cdots,h_o^{\left(N\right)}\right],\ h_o^{\left(i\right)}=\int_{0}^{\infty} x dF_H^{\left(i\right)}\left(x\right)$.
\item \textbf{Optimization step}: After the certainty step, the recursive  optimization problem in~\eqref{eqn6} at time slot $t$ can be reformulated as the following deterministic optimization
\begin{equation}\label{eqn8}
\begin{aligned}
\max_{\tilde{\mathbf{e}}_t,\cdots,\tilde{\mathbf{e}}_T} & R\left(\tilde{\mathbf{e}}_t,\mathbf{h}_t\right)+\sum_{k=t+1}^{T} R\left(\tilde{\mathbf{e}}_k,\mathbf{h}_o\right)\\
\text{s.t.} &\quad \sum_{k=t}^{T}\tilde{e}_k^{\left(i\right)}\leq \mathcal{E}_t^{\left(i\right)},\ i \in \mathcal{N}\\
&\quad\tilde{\mathbf{e}}_k \succeq \mathbf{0},\ k=t,\cdots,T
\end{aligned}
\end{equation}
where the solution of the optimization problem in~\eqref{eqn8} is obtained in a similar way to the offline allocation problem introduced in Section~\ref{offline} by applying the IWF algorithm in Algorithm~\ref{IWF} over $T-t+1$ slots with an amount of energy available at each user $\mathcal{E}_t^{\left(i\right)}$  for $i\in\mathcal{N}$.
\item \textbf{Allocation step}: We set $\mathbf{e}_t=\tilde{\mathbf{e}}_t$  and compute the energy levels of all users at $t+1$  using Equation~\ref{energylevel}. Then, we go to the next time slot $t+1$.
\end{enumerate}

\section{Numerical results}
\label{results}

\begin{figure}
\centering
\includegraphics[width=9 cm, height=7 cm]{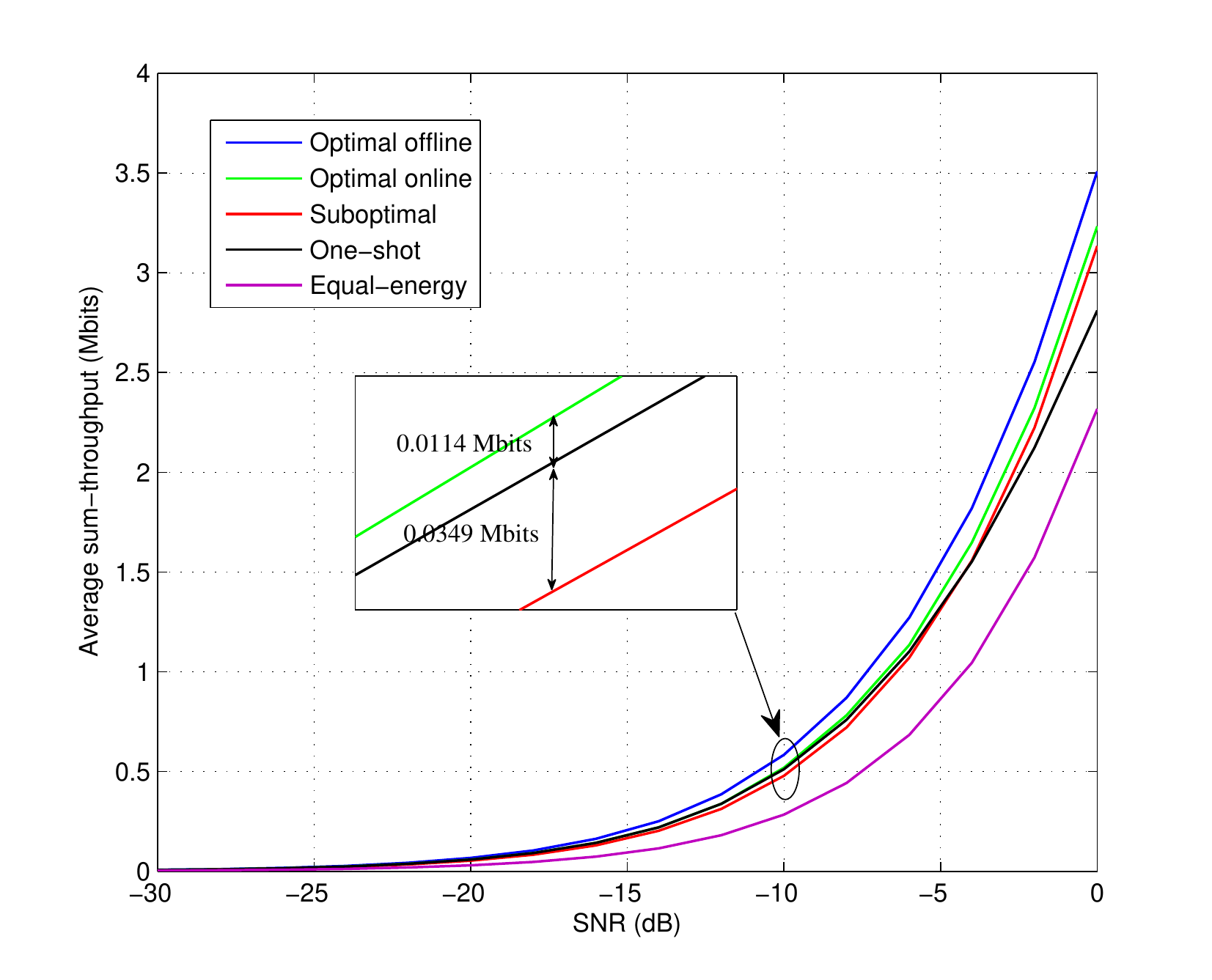}
\caption{Average sum-throughput in Mbits in the low SNR regime for $T=5$ and $N=2$.}
\label{fig1}
\end{figure}

In this section, we numerically evaluate the performance of various energy allocation policies introduced throughout the paper.  For comparison, we consider a simple energy allocation policy namely the equal-energy allocation, where each user allocates an equal amount of energy for each time slot of the transmission window regardless the effect of the channel fading and the allocation policy of the other users, i.e.,
\begin{equation}
e_{t}^{\left(i\right)}=\frac{E_i}{T},\ \forall i\in\mathcal{N},\ \forall t\in\mathcal{T}
\end{equation}
Notice that this policy is optimal in case of time-invariant channels, where the channel gain of each user is constant over the deadline.

\begin{figure}
\centering
\includegraphics[width=9 cm, height=7 cm]{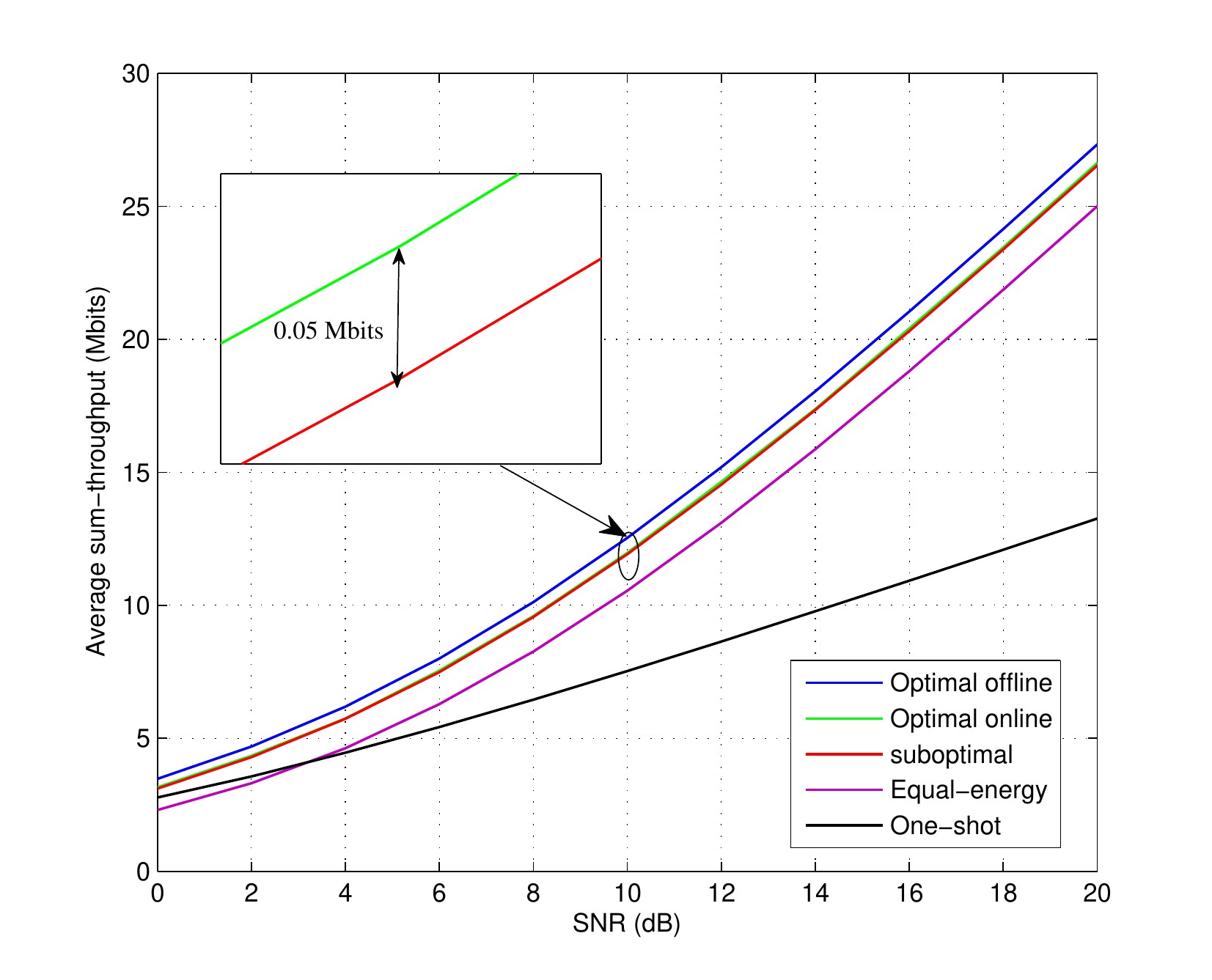}
\caption{Average sum-throughput in Mbits in the high SNR regime for $T=5$ and $N=2$}
\label{fig2}
\end{figure}

\begin{figure}
\centering
\begin{subfigure}{.5\textwidth}
\includegraphics[width=9 cm, height=7 cm]{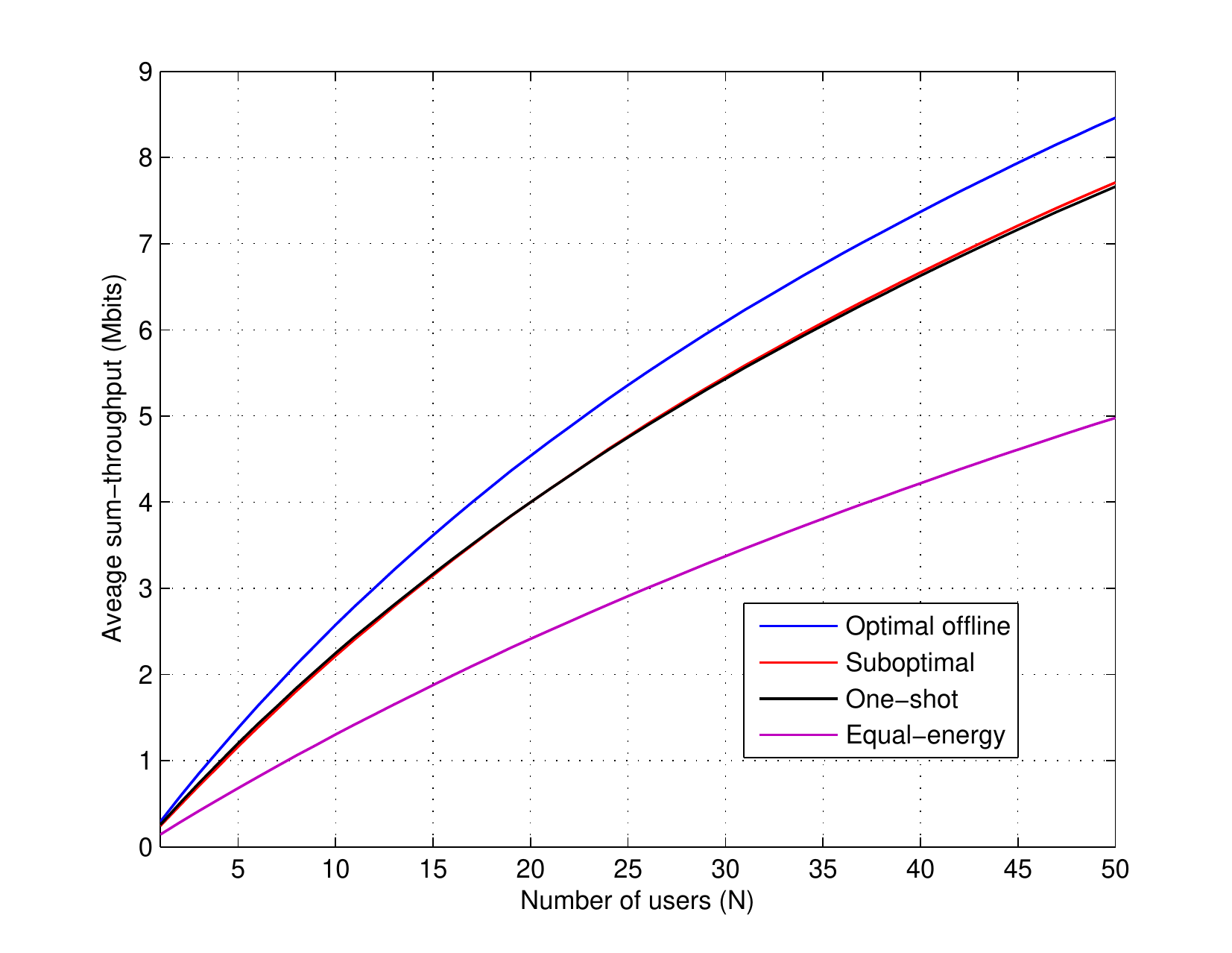}
\centering
\caption{$\text{SNR}=-10$ dB}
\label{fig3:1}
\end{subfigure}
\begin{subfigure}{.5\textwidth}
\includegraphics[width=9 cm, height=7 cm]{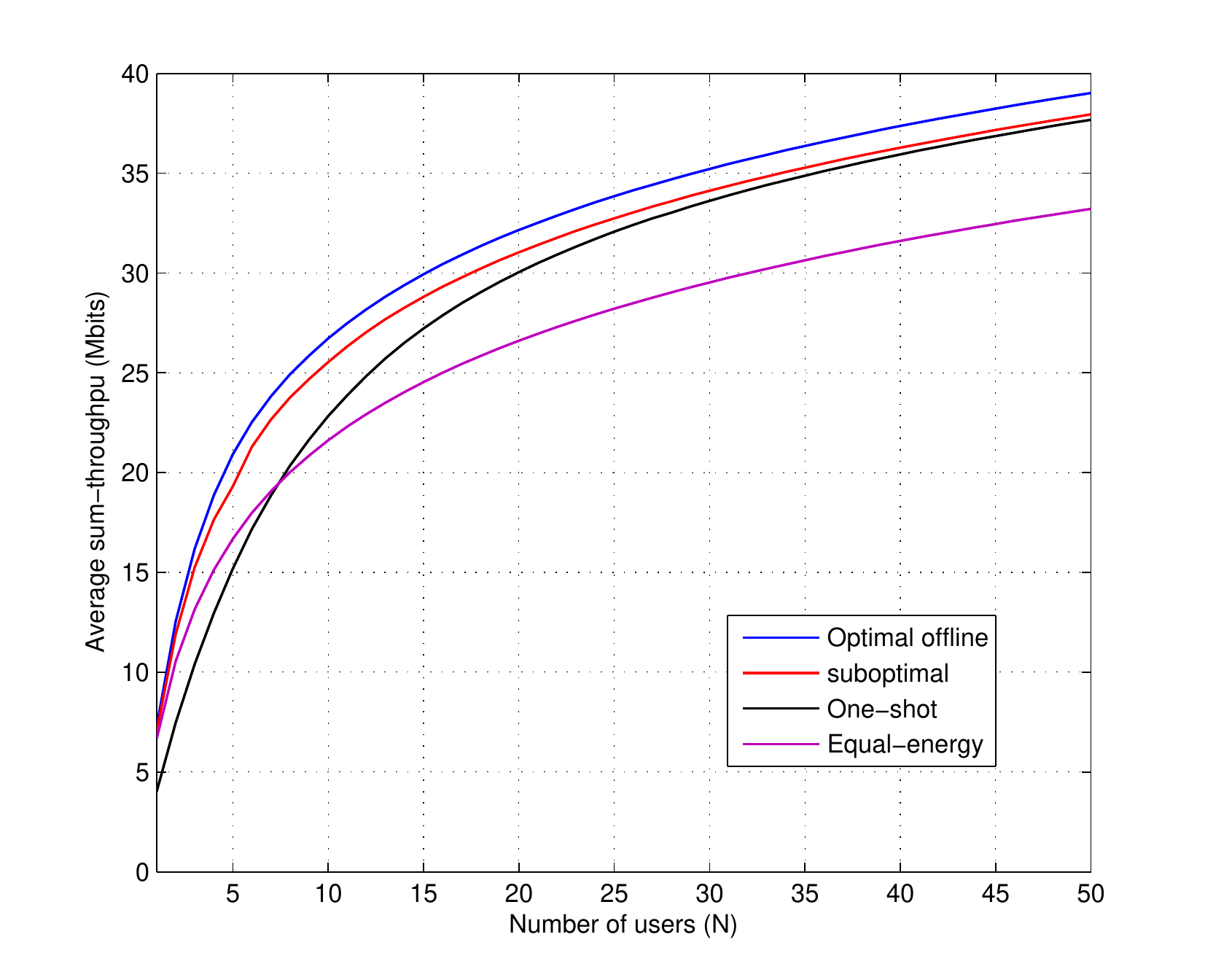}
\centering
\caption{$\text{SNR}=10$ dB}
\label{fig3:2}
\end{subfigure}
\caption{Average sum-throughput of the MAC versus the number of users $N$ for $\text{SNR}=-10$ dB and $\text{SNR}=10$ dB}
\end{figure}

For simplicity, we consider a symmetrical case, where all users are equipped with an equal amount of energy, i.e., $E_i=E$, $\forall i\in\mathcal{N}$, and the channel gains of all users are i.i.d., where the channel gains are generated according to the exponential distribution with parameter $\lambda=1$, i.e., $F_{{H}}^{\left(i\right)}\left(x\right)=1-e^{-x}$, $\forall i\in\mathcal{N}$. Also, we consider the following parameters: the bandwidth $W=1$ MHz, the noise power $N_o=1$ watts, and the slot length $\tau=1$ seconds, and hence, the transmit SNR of each user $\text{SNR}_i=E$, $\forall i\in\mathcal{N}$. We use the performance of the offline allocation policy as an upper bound on the performance of online policies. In the following figures, the performance of the optimal offline, suboptimal, one-shot, and equal-energy policies are obtained by averaging over $10^4$ randomly generated channel realizations, while the performance of the optimal online policy is obtained by using the discretization method~\cite{bertsekas1995dynamic}.

Figs~\ref{fig1} and~\ref{fig2} show the average sum-throughput of the MAC versus the transmit SNR of each user for a system composed of $N=2$ users and transmission window length equal to $T=5$ time slots. Fig.~\ref{fig1} focuses on the low SNR regime where the SNR is varied from $-30$ dB to $0$ dB. It is clear that the performance of the proposed suboptimal and the one-shot policies is close to the optimal one, although, the proposed suboptimal policy performs better when the SNR approaches $0$ dB. Moreover, the equal-energy allocation policy has the worst performance. In Fig.~\ref{fig2}, the SNR varies from $0$ dB to $20$ dB to investigate the performance of the different policies in the medium and high SNR regimes. We can see from this figure that the one-shot policy deviates from the optimal solution, since the linear approximation of the throughput function is no longer valid at high SNR. However, the performance of the proposed suboptimal policy is still very close to that of the optimal solution.

Next, we investigate the effect of the number of users $N$ on the performance of different policies. Fig.~\ref{fig3:1} and Fig.~\ref{fig3:2} show the average sum throughput of the system (for $T=5$ slots) at SNR$=-10$ dB and SNR=$10$ dB, respectively. We can see from Fig.~\ref{fig3:1} that both the proposed suboptimal policy and the one-shot policy almost have the same performance for any number of users in the low SNR regime. 

When the number of users is much larger than the time slots of the transmission window, i.e., $N\gg T$, each time slot of the transmission widow would be shared with a lot of users. In other words, each user would suffer from high interference signals at each time slot of the transmission window. Therefore the best choice is to allocate the available energy of each user to a single time slot of the transmission window that has a favorable channel gain. Hence Fig.~\ref{fig3:2} shows that the one-shot policy converges to the proposed suboptimal policy in the high SNR regime for $N\gg T$. However, the equal-energy allocation policy has better performance than the one-shot policy when the number of users is small. On the other hand, Fig.~\ref{fig3:1} and Fig.~\ref{fig3:2} show that the gap between the equal-energy allocation policy and the suboptimal policy increases as the number of users increases since the competition on the available resources (the time slots of the transmission window) increases as the number of users increases.

\section{Conclusion}
\label{conclusion}

In this paper, we have proposed energy allocation strategies for the $N$-user fading MAC with delay and energy constraints under two different assumptions on the channel states information. In the offline allocation, a convex optimization problem is formulated with the objective of maximizing the sum-throughput of the fading MAC within the transmission window where the optimal solution is obtained by applying the iterative water filling algorithm. In the online allocation, the problem is formulated via dynamic programming, and the optimal solution is obtained numerically by using the discretization method when the number of users is small. In addition, we have proposed a suboptimal solution with reduced computational complexity that can be used when the number of users is large. Numerical results have been provided to show the superiority of the proposed algorithms compared to the equal-energy allocation and the one-shot allocation algorithms.

\appendix
The optimal offline transmission policy is obtained by solving the optimization problem~\eqref{eqn2}. Since the objective function of~\eqref{eqn2} is the sum of concave functions $R\left(\mathbf{e},\mathbf{h}\right)$, and the constrains are affine functions, then the optimization problem~\eqref{eqn2} is a convex optimization problem that can be solved using Lagrange method. The Lagrangian is given by
\begin{equation}
\mathcal{L}=\sum_{t=1}^{T} R\left(\mathbf{e}_t,\mathbf{h}_t\right)-\sum_{i=1}^{N}\mu^{\left(i\right)}\left(\sum_{t=1}^{T}e_t^{\left(i\right)}-E_i\right)-\sum_{i=1}^{N}\sum_{t=1}^{T}\lambda_{t}^{\left(i\right)}e_{t}^{\left(i\right)}
\end{equation}
where $\mu^{\left(i\right)}$ is the Lagrange multiplier associated with the $i$th equality constraint in~\eqref{eqn1}, and $\lambda_t^{\left(i\right)}$ is the Lagrange multiplier associated with the $i$th inequality constraint in~\eqref{eqn1}. Slater’s condition is satisfied for this problem, and hence, the Karush-Kuhn-Tucker (KKT) conditions are necessary and sufficient for optimality~\cite{boyd2004convex}. The KKT conditions are given by
\begin{subequations}\label{eqn5}
\begin{align}
\frac{\partial\mathcal{L}}{\partial e_t^{\left(i\right)}}&= \frac{ \tau W h_t^{\left(i\right)}\ln\left(2\right)}{\tau N_o+\sum_{k=1}^{N}h_t^{\left(k\right)}e_t^{\left(k\right)}}-\mu^{\left(i\right)}-\lambda_t^{\left(i\right)}\ \: \forall i\in\mathcal{N},\ t\in\mathcal{T}\\
\label{eqn5:2}
\lambda_t^{\left(i\right)}e_t^{\left(i\right)}&=0,\quad \lambda_t^{\left(i\right)}\geq 0,\quad e_t^{\left(i\right)}\geq 0,\quad \forall i\in\mathcal{N},\ t\in\mathcal{T}\\
\sum_{t=1}^{T}e_t^{\left(i\right)}&=E_i\ ,\quad \forall i\in\mathcal{N}
\end{align}
\end{subequations}
Solving the above KKT conditions yields \eqref{eqn3}--\eqref{last_IWF}.

\nocite{*}
\bibliographystyle{IEEEtran}

\end{document}